\documentclass[preprint,12pt]{elsarticle}

\usepackage{graphicx}

\usepackage{amssymb}

\usepackage{amssymb, amsmath, amsthm}

\newcommand{\commentOut}[1]{}
\newtheorem{theorem}{Theorem}

\newtheorem{lemma}{Lemma}

\journal{}

\begin{document}

\begin{frontmatter}



\title{An Elegant Algorithm for the Construction of Suffix Arrays}
 
 
\author[raj]{Sanguthevar Rajasekaran}
\ead{rajasek@engr.uconn.edu}
\author[marius]{Marius Nicolae}
\ead{marius.nicolae@engr.uconn.edu}


\address{Dept. of Computer Science and Engineering, Univ. of Connecticut,
    Storrs, CT, USA}


\begin{abstract}
        The suffix array is a data structure that finds numerous applications in string 
        processing problems for both linguistic texts and biological data. It has been 
        introduced as a memory efficient alternative for suffix trees. The suffix array 
        consists of the sorted suffixes of a string. There
        are several linear time suffix array construction algorithms (SACAs) known in the literature.
        However, one of the fastest algorithms in practice has a worst case
        run time of $O(n^2)$. The problem of designing practically and  theoretically 
        efficient techniques remains open.
        
        In this paper we present an elegant
        algorithm for suffix array construction which takes
        linear time with high probability; the probability is on the space of all possible
        inputs. Our algorithm is one of the simplest of the known SACAs and it opens up 
        a new dimension of suffix array construction that has not been explored until now. 
        Our algorithm is easily parallelizable. We offer parallel implementations on various parallel
        models of computing. 
        We prove a lemma on the $\ell$-mers of a random string which might find
        independent applications. 
        We also present another algorithm that utilizes 
        the above algorithm. This algorithm is called RadixSA and  
        has a worst case run time of  $O(n\log{n})$. RadixSA introduces an idea that may find independent
        applications as a speedup technique for other SACAs. 
        An empirical comparison of RadixSA with other algorithms
        on various datasets reveals that our algorithm is one of the fastest algorithms to date.
        The C++ source code is freely available at
        {http://www.engr.uconn.edu/$\sim$man09004/radixSA.zip}.
        
\end{abstract}


\end{frontmatter}



\section{Introduction}

The suffix array is a data structure that finds numerous applications in string
processing problems for both linguistic texts and biological data. It has been
introduced in \cite{MaMy93} as a memory efficient alternative to suffix trees. 
The suffix array of
a string $T$ is an array $A$, ($|T|=|A|=n$) which gives the lexicographic order
of all the suffixes of $T$. Thus, $A[i]$ is the starting position of
the lexicographically $i$-th smallest suffix of $T$.

The original suffix array construction
algorithm \cite{MaMy93} runs in $O(n\log n)$ time. It
is based on a technique called {\em prefix doubling:} assume that the suffixes
are grouped into buckets such that suffixes 
in the same bucket share
the same prefix of length $k$. Let $b_i$ be the bucket
number for suffix $i$. Let $q_i=(b_i, b_{i+k})$. Sort the suffixes with
respect to $q_i$ using radix sort. As a result, the suffixes become sorted by their first $2k$
characters. Update the bucket numbers and repeat the process until all the
suffixes are in buckets of size $1$. This process takes no more than $\log n$
rounds.
The idea of sorting suffixes in one bucket based on the bucket
information of nearby suffixes is called {\em induced copying}. It appears in some form or
another in many of the algorithms for suffix array
construction.

Numerous papers have been written on suffix arrays. A survey
on some of these algorithms can be found in \cite{PST2007}. The authors of
\cite{PST2007} categorize suffix array construction algorithms  (SACA) into
five based on the main techniques employed: 1) Prefix Doubling (examples include
\cite{MaMy93} - run time  $=O(n\log n)$; \cite{LaSa99} - run time $=O(n\log
n)$); 2) Recursive (examples include \cite{KJP04} - run time $=O(n\log\log
n)$); 3) Induced Copying (examples include  \cite{BaBr05} - run time $=O(n\sqrt{\log n})$); 
4) Hybrid (examples
include \cite{ItTa99} and \cite{KoAl03} - run time $=O(n^2\log n)$); and 5)
Suffix Tree (examples include \cite{Kur99} - run time $=O(n\log \sigma)$ where
$\sigma$ is the size of the alphabet).

In 2003, three independent groups \cite{KoAl03,KaSa03,KSP+03} found the first
linear time suffix array construction algorithms which do not require building a
suffix tree beforehand. For example, in \cite{KoAl03} the suffixes
are classified as either $L$ or $S$. Suffix $i$ is an $L$ suffix if it is lexicographically larger than
suffix $i+1$, otherwise it is an $S$ suffix. Assume that the number of $L$
suffixes is less than $n/2$, if not, do this for $S$ suffixes. Create a new
string where the segments of text in between $L$ suffixes are renamed
to single characters. The new text has length no more than $n/2$ 
and we recursively find its suffix array. This suffix array gives the order of 
the $L$ suffixes in the original string. This order is used to induce the order 
of the remaining suffixes.

Another linear time algorithm, called {\em skew}, is given
in \cite{KaSa03}. It first sorts those suffixes $i$ with $i \;{\bf mod}\; 3 \neq
0$ using a recursive procedure. The order of these suffixes is then used to
infer the order of the suffixes with $i
\;{\bf mod}\; 3 = 0$. Once these two groups are determined we can compare
one suffix from the first group with one from the second group in constant time.
The last step is to merge the two sorted groups, in linear time.

Several other SACAs have been proposed in the literature in recent
years (e.g., \cite{NZC09,ScSt07}). Some of the algorithms with superlinear
worst case run times perform better in practice than the linear ones.
One of the currently best performing algorithms in practice is the $BPR$ 
algorithm of \cite{ScSt07} which
has an asymptotic worst-case run time of $O(n^2)$. $BPR$ first sorts all the
suffixes up to a certain depth, then focuses on one bucket at a time and
repeatedly refines it into sub-buckets.

In this paper we present an elegant algorithm for suffix array construction.
This algorithm takes linear time with high probability. Here the probability is 
on the space of all possible inputs. Our algorithm is one of the simplest 
algorithms known for constructing suffix arrays. It opens up a new dimension in
suffix array construction, i.e., the development of algorithms with provable
expected run times. This dimension has not been explored before. We prove a
lemma on the $\ell$-mers of a random  string which might find independent
applications. Our algorithm is  also nicely parallelizable. We offer parallel
implementations of our algorithm on various parallel models of computing.

We also present another algorithm for suffix array construction that utilizes 
the above algorithm. This algorithm, called RadixSA, is based on bucket sorting
and has a worst case run time of $O(n\log{n})$. It employs an idea which, 
to the best of our knowledge, has not been directly exploited until now. 
RadixSA selects the order in which buckets are processed based on a heuristic
such that, downstream, they impact as many other buckets as possible. 
This idea may find independent application as a standalone speedup technique for 
other SACAs based on bucket sorting.
RadixSA also employs a generalization of Seward's copy method \cite{Sew2000} 
(initially described in \cite{BuWe94}) to
detect and handle repeats of any length. We compare RadixSA with other algorithms 
on various datasets.

\section{A Useful Lemma}

Let $\Sigma$ be an alphabet of interest and let $S=s_1s_2\ldots s_n\in
\Sigma^*$. Consider the case when $S$ is generated randomly, i.e., each $s_i$ is
picked uniformly randomly from $\Sigma$ ($1\leq i\leq n$). Let $L$ be the set of
all $\ell$-mers of $S$. Note that $|L|=n-\ell+1$. What can we say about the
independence of these $\ell$-mers? In several papers analyses have been done
assuming that these $\ell$-mers are independent (see e.g., \cite{BuTo00}). These authors point out that
this assumption may not be true but these analyses have proven to be useful in
practice. In this Section we prove the following Lemma on these $\ell$-mers.

\begin{lemma}\label{rslemma} Let $L$ be the set of all $\ell$-mers of a random
string generated from an alphabet $\Sigma$. Then, the $\ell$-mers in $L$ are
pairwise independent. These $\ell$-mers need not be $k$-way independent for
$k\geq 3$.  \end{lemma}

\begin{proof} Let $A$ and $B$ be any two $\ell$-mers in $L$. If $x$ and $y$ are
non-overlapping, clearly, $Prob[A=B]=(1/\sigma)^{\ell}$, where
$\sigma=|\Sigma|$. Thus, consider the case when $x$ and $y$ are overlapping.

Let $P_i=s_is_{i+1}\ldots s_{i+\ell-1}$, for $1\leq i\leq (n-\ell+1)$. Let
$A=P_i$ and $B=P_j$ with $i< j$ and $j\leq (i+\ell-1)$. Also let $j=i+k$ where
$1\leq k\leq (\ell-1)$.

Consider the special case when $k$ divides $\ell$. If $A=B$, then it should be
the case that $s_i=s_{i+k}=s_{i+2k}=\cdots =s_{i+\ell}$;
$s_{i+1}=s_{i+k+1}=s_{i+2k+1}=\cdots=s_{i+\ell+1}$; $\cdots$; and
$s_{i+k-1}=s_{i+2k-1}=s_{i+3k-1}=\cdots=s_{i+\ell+k-1}$. In other words, we
have $k$ series of equalities. Each series is of length $(\ell/k)+1$. The
probability of all of these equalities is $\left (\frac{1}{\sigma}\right
)^{\ell/k} \left (\frac{1}{\sigma}\right )^{\ell/k}$ $\cdots$  $\left
(\frac{1}{\sigma}\right )^{\ell/k}$ $=\left (\frac{1}{\sigma}\right )^{\ell}$.

As an example, let $S=abcdefghi$, $\ell=4$, $k=2$, $A=P_1$, and $B=P_3$. In
this case, the following equalities should hold: $a=c=e$ and $b=d=f$. The
probability of all of these equalities is $
(1/\sigma)^2(1/\sigma)^2=(1/\sigma)^4=(1/\sigma)^\ell$.

Now consider the general case (where $k$ may not divide $\ell$). Let
$\ell=qk+r$ for some integers $q$ and $r$ where $r<k$. If $A=B$, the following
equalities will hold: $s_i=s_{i+k}=s_{i+2k}=\cdots
=s_{i+\lfloor(\ell+k-1)/k\rfloor k}$;
$s_{i+1}=s_{i+1+k}=s_{i+1+2k}=\cdots=s_{i+1+\lfloor(\ell+k-2)/k\rfloor k}$;
$\cdots$; and
$s_{i+k-1}=s_{i+k-1+k}=s_{i+k-1+2k}=\cdots=s_{i+k-1+\lfloor(\ell/k)\rfloor k}$.

Here again we have $k$ series of equalities. The number of elements in the
$q$th series is $1+\left\lfloor\frac{\ell+k-q}{k}\right\rfloor$, for $1\leq
q\leq k$. The probability of all of these equalities is $(1/\sigma)^x$ where
$x=\sum_{q=1}^k\left\lfloor\frac{\ell+k-q}{k}\right\rfloor$.
$$x=\left\lfloor\frac{(q+1)k+r-1}{k}\right\rfloor+\left\lfloor\frac{(q+1)k+r-2}{k}\right\rfloor+\cdots+\left\lfloor\frac{(q+1)k}{k}\right\rfloor$$
$$+
\left\lfloor\frac{(q+1)k-1}{k}\right\rfloor+\left\lfloor\frac{(q+1)k-2}{k}\right\rfloor+\cdots+\left\lfloor\frac{(q+1)k-(k-r)}{k}\right\rfloor$$
$$=(q+1)r+(k-r)q=kq+r=\ell.$$

The fact that the $\ell$-mers of $L$ may not be $k$-way independent for $k\geq
3$ is easy to see. For example, let $S=abcdefgh$, $\ell=3$, $A=P_1$, $B=P_3$,
and $C=P_4$. What is $Prob.[A=B=C]$? If $A=B=C$, then it should be the case
that $a=c,b=d=a,b=c=e$, and $c=f$. In other words, $a=b=c=d=e=f$. The
probability of this happening is $(1/\sigma)^5\neq(1/\sigma)^6$.  \end{proof}

\noindent{\bf Note:} To the best of our knowledge, the above lemma cannot be
found in the existing literature. In \cite{Szp01} a lemma is proven on the
expected depth of insertion of a suffix tree. If anything, this only very
remotely resembles our lemma but is not directly related. In addition the lemma
in \cite{Szp01} is proven only in the limit (when $n$ tends to $\infty$).

\section*{Our Basic Algorithm}
 Let $S=s_1s_2\cdots s_n$ be the given input string. Assume that $S$ is a
string randomly generated from an alphabet $\Sigma$. In particular, each $s_i$
is assumed to have been picked uniformly randomly from $\Sigma$ (for $1\leq
i\leq n$). For all the algorithms presented in this paper, no assumption is
made on the size of $\Sigma$. In particular, it could be anything. For example,
it could be $O(1)$, $O(n^c)$ (for any constant $c$), or larger.

The problem is to produce an array $A[1:n]$ where $A[i]$ is the starting
position of the $i$th smallest suffix of $S$, for $1\leq i\leq n$. The basic
idea behind our algorithm is to sort the suffixes only with respect to their
prefixes of length $O(\log n)$ (bits). The claim is that this amount of sorting is
enough to order the suffixes with {\em high probability}. By high probability
we mean a probability of $\geq (1-n^{-\alpha})$ where $\alpha$ is the
probability parameter (typically assumed to be a constant $\geq 1$). The
probability space under concern is the space of all possible inputs.

Let $S_i$ stand for the suffix that starts at position $i$, for $1\leq i\leq
n$. In other words, $S_i=s_is_{i+1}\cdots s_n$. Let $P_i=s_is_{i+1}\cdots
s_{i+\ell-1}$, for $i\leq (n-\ell)$. When $i>(n-\ell)$, let $P_i=S_i$. The
value of $\ell$ will be decided in the analysis. A pseudocode of our basic
algorithm follows.


\begin{tabbing}
aaaa \= aaaa \= aaaa \= aaaa \kill
{\bf Algorithm SA1}\\
\> Sort $P_1,P_2,\ldots,P_n$ using radix sort;\\
\> The above sorting partitions the $P_i$'s into buckets where equal
$\ell$-mers\\ 
\> are in the same bucket;\\
\> Let these buckets be $B_1,B_2,\ldots,B_m$ where $m\leq n$;\\
\> {\bf for} $i:=1$ {\bf to} $m$ {\bf do}\\
\>\> {\bf if} $|B_i|>1$ {\bf then} sort the suffixes corresponding to the
$\ell$-mers in $B_i$\\
\>\>\> using any relevant algorithm;
\end{tabbing}

 
\begin{lemma}
Algorithm SA1 has a run time of $O(n)$ with high probability.
\end{lemma} 

\begin{proof}
Consider a specific $P_i$ and let $B$ be the bucket that $P_i$ belongs to after
the radix sorting step in Algorithm SA1. How many other $P_j$'s will there be
in $B$? Using Lemma~\ref{rslemma}, $Prob.[P_i=P_j]=(1/\sigma)^\ell$. This means
that $Prob.[\exists j: i\neq j \& P_i=P_j]\leq n(1/\sigma)^\ell$. As a result,
$Prob.[\exists j: |B_j|>1]\leq n^2(1/\sigma)^\ell$. If
$\ell\geq((\alpha+2)\log_\sigma n)$, then, $n^2(1/\sigma)^\ell\leq
n^{-\alpha}$.

In other words, if $\ell\geq((\alpha+2)\log_\sigma n)$, then each bucket will
be of size 1 with high probability. Also, the radix sort will take $O(n)$ time. Note that we only need to sort $O(\log n)$ bits of each $P_i$ ($1\leq i\leq n$) and this sorting can be done in $O(n)$ time (see e.g., \cite{HSR08}).
\end{proof}


\noindent{\bf Observation 1.} We could have a variant of the algorithm where if
any of the buckets is of size greater than 1, we abort this algorithm and use
another algorithm. A pseudocode follows.


\begin{tabbing}
aaaa \= aaaa \= aaaa \= aaaa \= aaaa \kill
{\bf Algorithm SA2}\\
{\bf 1}\> Sort $P_1,P_2,\ldots,P_n$ using radix sort;\\
\> The above sorting partitions the $P_i$'s into buckets where equal $\ell$-mers \\
\> are in the same bucket; \\
\> Let these buckets be $B_1,B_2,\ldots,B_m$ where $m\leq n$;\\
{\bf 2}\> {\bf if} $|B_i|=1$ for each $i,1\leq i\leq m$\\
{\bf 3}\> {\bf then} output the suffix array and quit;\\
{\bf 4}\> {\bf else} use another algorithm (let it be {\bf Algorithm SA}) to
find and \\
\>\>output the suffix array;
\end{tabbing}


\noindent{\bf Observation 2.} Algorithm SA1 as well as Algorithm SA2 run in
$O(n)$ time on at least $(1-n^{-\alpha})$ fraction of all possible inputs.
Also, if the run time of Algorithm SA is $t(n)$, then the expected run time of
Algorithm SA2 is $(1-n^{-\alpha})O(n)+n^{-\alpha}(O(n)+t(n))$. For example, if
Algorithm SA is the skew algorithm \cite{KaSa03}, then the expected run time
of Algorithm SA2 is $O(n)$ (the underlying constant will be smaller than the
constant in the run time of skew).

\vspace{0.1in}

\noindent{\bf Observation 3.} In general, if $T(n)$ is the run time of
Algorithm SA2 lines 1 through 3 and if $t(n)$ is the run time of Algorithm SA,
then the expected run time of Algorithm SA2 is
$(1-n^{-\alpha})T(n)+n^{-\alpha}(T(n)+t(n))$.

\vspace{0.1in}

\noindent{\bf The case of non-uniform probabilities.} In the above algorithm
and analysis we have assumed that each character in $S$ is picked uniformly
randomly from $\Sigma$. Let $\Sigma=\{a_1,a_2,\ldots, a_\sigma\}$. Now we
consider the possibility that for any $s_i\in S$, $Prob.[s_i=a_j]=p_j$, $1\leq
i\leq n;1\leq j\leq \sigma$. For any two $\ell$-mers $A$ and $B$ of $S$ we can
show that $Prob.[A=B]=\left (\sum_{j=1}^\sigma p_j^2\right )^\ell$. In this
case, we can employ Algorithms SA1 and SA2 with
$\ell\geq(\alpha+2)\log_{1/P}n$, where $P=\sum_{j=1}^\sigma p_j^2$.

\vspace{0.2in}

\noindent{\bf Observation 4.} Both SA1 and SA2 can work with any alphabet size.
If the size of the alphabet is $O(1)$, then each $P_i$ will consist of $O(\log
n)$ characters from $\Sigma$. If $|\Sigma|=\Theta(n^c)$ for some constant $c$,
then $P_i$ will consist of $O(1)$ characters from $\Sigma$. If
$|\Sigma|=\omega(n^c)$ for any constant $c$, then each $P_i$ will consist of a
prefix (of length $O(\log n)$ bits) of a character in $\Sigma$.

\section{Parallel Versions}
In this Section we explore the possibility of implementing SA1 and SA2 on various models of parallel computing. 

\subsection{Parallel Disks Model}
In a Parallel Disks Model (PDM), there is a (sequential or parallel) computer
whose core memory is of size $M$.  The computer has $D$ parallel disks. In one
parallel I/O, a block of size $B$ from each of the $D$ disks can be fetched
into the core memory. The challenge is to devise algorithms for this model that
perform the least number of I/O operations. This model has been proposed to
alleviate the I/O bottleneck that is common for single disk machines especially
when the dataset is large. In the analysis of PDM algorithms the focus is on
the number of parallel I/Os and typically the local computation times are not
considered. A lower bound on the number of parallel I/Os needed to sort $N$
elements on a PDM is $\frac{N}{DB}\frac{\log(N/B)}{\log(M/B)}$. Numerous
asymptotically optimal parallel algorithms have been devised for sorting on the
PDM. For practical values of $N,M,D,$ and $B$, the lower bound basically means
a constant number of passes through the data. Therefore, it is imperative to
design algorithms wherein the underlying constants in the number of I/Os is
small. A number of algorithms for different values of $N,M,D,$ and $B$ that
take a small number of passes have been proposed in \cite{RaSe08}. 

One of the algorithms given in \cite{RaSe08} is for sorting integers. In
particular it is shown that we can sort $N$ random integers in the range
$[1,R]$  (for any $R$) in $(1+\nu)\frac{\log(N/M)}{\log(M/B)}+1$ passes through
the data, where $\nu$ is a constant $<1$. This bound holds with probability
$\geq (1-N^{-\alpha})$, this probability being computed in the space of all
possible inputs.

We can adapt the algorithm of \cite{RaSe08} for constructing suffix arrays as
follows. We assume that the word length of the machine is $O(\log n)$. This is
a standard assumption made in the algorithms literature. Note that if the
length of the input string is $n$, then we need a word length of at least $\log
n$ to address the suffixes. To begin with, the input is stored in the $D$ disks
striped uniformly. We generate all the $\ell$-mers of $S$ in one pass through
the input. Note that each $\ell$-mer occupies one word of the machine. The
generated $\ell$-mers are stored back into the disks. Followed by this, these
$\ell$-mers are sorted using the algorithm of \cite{RaSe08}. At the end of this
sorting, we have $m$ buckets where each bucket has equal $\ell$-mers. As was
shown before, each bucket is of size 1 with high probability. 

We get the following:
\begin{theorem}
We can construct the suffix array for a random string of length $n$  in
$(1+\nu)\frac{\log(n/M)}{\log(M/B)}+2$ passes through the data, where $\nu$ is
a constant $<1$. This bound holds for $\geq (1-n^{-\alpha})$ fraction of all
possible inputs. $\Box$ 
\end{theorem} 

\subsection{The Mesh and the Hypercube}
Optimal algorithms exist for sorting on interconnection networks such as the
mesh (see e.g., \cite{ThKu77} and \cite{KKNT91}), the hypercube (see e.g.,
\cite{ReVa87}), etc. We can use these in conjunction with Algorithms SA1 and
SA2 to develop suffix array construction algorithms for these models. Here
again we can construct all the $\ell$-mers of the input string. Assume that we
have an interconnection network with $n$ nodes and each node stores one of the
characters in the input string. In particular node $i$ stores $s_i$, for $1\leq
i\leq n$. Depending on the network, a relevant indexing scheme has to be used.
For instance, on the mesh we can use a snake-like row-major indexing. Node $i$
communicates with nodes $i+1,i+2,\ldots,i+\ell-1$ to get
$s_{i+1},s_{i+2},\ldots,s_{i+\ell-1}$. The communication time needed is $O(\log
n)$. Once the node $i$ has these characters it forms $P_i$. Once the nodes have
generated the $\ell$-mers, the rest of the algorithm is similar to the
Algorithm SA1 or SA2. As a result, we get the following:

\begin{theorem}
There exists a randomized algorithm for constructing the suffix array for a
random string of length $n$ in $O(\log n)$ time on a $n$-node hypercube with
high probability. The run time of \cite{ReVa87}'s algorithm is $O(\log n)$ with
high probability, the probability being computed in the space of all possible
outcomes for the coin flips made. Also, the same can be done in $O(\sqrt n)$
time on a $\sqrt n\times\sqrt n$ mesh with high probability. $\Box$
\end{theorem}

\vspace{0.1in}

\noindent{\bf Observation:} Please note that on a $n$-node hypercube, sorting
$n$ elements will need $\Omega(\log n)$ time even if these elements are bits,
since the diameter of the hypercube is $\Omega(\log n)$. For the same reason,
sorting $n$ elements on a $\sqrt n\times\sqrt n$ mesh will need $\Omega(\sqrt
n)$ time since $2(\sqrt n-1)$ is the diameter.



\subsection{PRAM Algorithms}
In \cite{KaSa03} several PRAM algorithms are given. One such algorithm is for
the EREW PRAM that has a run time of $O(\log^2n)$, the work done being $O(n\log
n)$. We can implement Algorithm SA2 on the EREW PRAM so that it has an expected
run time of $O(\log n)$, the expected work done being $O(n\log n)$. Details
follow. Assume that we have $n$ processors. 1) Form all possible $\ell$-mers.
Each $\ell$-mer occupies one word; 2) Sort these $\ell$-mers using the parallel
merge sort algorithm of \cite{Col88};  3) Using a prefix computation check if
there is at least one bucket of size $>1$; 4) Broadcast the result to all the
processors using a prefix computation; 5) If there is at least one bucket of
size more than one, use the parallel algorithm of \cite{KaSa03}.

Steps 1 through 4 of the above algorithm take $O(\log n)$ time each. Step 5
takes $O(\log^2n)$ time. From Observation 3, the expected run time of this
algorithm is $(1-n^{-\alpha})O(\log n)+n^{-\alpha}(O(\log n)+O(\log^2n))=O(\log
n)$. Also, the expected work done by the algorithm is $(1-n^{-\alpha})O(n\log
n)+n^{-\alpha}(O(n\log n)+O(n\log^2n))=O(n\log n)$.

\section{Practical Implementation}
In this section we discuss the design and implementation of the RadixSA algorithm.
The following is the pseudocode of the RadixSA algorithm:

\begin{tabbing}
aaaa \= aaaa \= aaaa \= aaaa \= aaaa \= aaaa \kill
{\bf Algorithm RadixSA}\\
{\bf 1.}\> {\bf radixSort} all suffixes by $d$ characters \\
{\bf 2.}\> {\bf let} $b[i]=$ bucket of suffix $i$ \\
{\bf 3.}\> {\bf for} $i:=n$ {\bf down to} $1$ {\bf do}\\
{\bf 4.}\>\> {\bf if (}$b[i].size > 1${\bf) then} \\
{\bf 5.}\>\>\> {\bf if detectPeriods}($b[i]$){\bf then}\\
{\bf 6.}\>\>\>\> {\bf handlePeriods}($b[i]$);\\
{\bf 7.}\>\>\> {\bf else} \>{\bf radixSort} all suffixes $j \in b[i]$ with
respect to $b[j+d]$\\
\end{tabbing}

A bucket is called {\em singleton} if it contains only one suffix, otherwise
it is called {\em non-singleton}. A
{\em singleton suffix} is the only suffix in a singleton bucket. A
singleton suffix has its final position in the suffix array already determined.

We number the buckets such that two suffixes in different buckets can be compared
simply by comparing their bucket numbers. The {\bf for} loop traverses the suffixes 
from the last to the first position in the text. This order ensures that
after each step, suffix $i$ will be found in a singleton bucket. This is easy
to prove by induction. Thus, at the end of the loop, all
the buckets will be singletons. If each bucket is of size $O(1)$ before the
{\bf for} loop is entered, then it is easy to see that the algorithm runs in
$O(n)$ time. 

Second, even if the buckets are not of constant size (before the
{\bf for} loop is entered) the algorithm is still linear if every suffix takes
part in no more than a constant number of radix sort operations. For that to happen, 
we want to give priority to buckets which can influence as many downstream buckets 
as possible. Intuitively, say a pattern $P$
appears several times in the input. We want to first sort the buckets which contain the
suffixes that start at the tails of the pattern instances. Once these suffixes are 
placed in different buckets we progress towards the buckets containing the heads of
the pattern instances. This way, every suffix is placed in a singleton bucket at a constant
cost per suffix. Our traversal order gives a good approximation of this behavior in 
practice, as we show in the results section.

Table \ref{table_example} shows an example of how the algorithm works. Each
column illustrates the state of the suffix array after sorting one of the
buckets. The order in which buckets are chosen to be sorted follows the pseudocode of RadixSA.
The initial radix sort has depth $1$ for illustration purpose. The last column
contains the fully sorted suffix array.

However, the algorithm as is described above has a worst case runtime of $O(n\sqrt{n})$ (proof omitted).  
We can improve the runtime to $O(n\log n)$ as follows. If, during the for loop, a bucket contains 
suffixes which have been accessed more than a constant $C$ number of times, we skip that bucket.
This ensures that the for loop takes linear time.
If at the end of the loop there have been any buckets skipped, we do another pass
of the for loop. After each pass, every remaining non-singleton bucket has
a sorting depth at least $C+1$ times greater than in the previous round (easy to prove
by induction). Thus, no more than a logarithmic number of passes will be needed and so
the algorithm has worst case runtime $O(n\log n)$.

\begin{table}
\label{table_example}
\caption{Example of suffix array construction steps for string `cdaxcdayca`. b[{\em
 suffix}] stands for the bucket of {\em suffix}. Underlines show the depth of
 sorting in a bucket at a given time. The
initial radix sort has depth $1$ for illustration purpose.}
\begin{tabular}{| l | l | l | l | l |}
\hline
{ Initial buckets} &{ Sort b[a]} & { Sort b[ca]} &{ Sort b[dayca]}    &{ Sort b[cdayca]} \\
\hline
{\underline{a}}& a & a & a & a \\
\cline{2-5}
{\underline{a}}yca& axcdayca & axcdayca & axcdayca & axcdayca \\
\cline{2-5}
{\underline{a}}xcdayca& ayca & ayca & ayca & ayca \\
\hline
{\underline{c}}a&{\underline{c}}a& ca & ca & ca \\
\cline{3-5}
{\underline{c}}dayca&{\underline{c}}dayca&{\underline{cd}}ayca&{\underline{cd}}ayca& cdaxcdayca \\
\cline{5-5}
{\underline{c}}daxcdayca&{\underline{c}}daxcdayca&{\underline{cd}}axcdayca&{\underline{cd}}axcdayca& cdayca\\
\hline
{\underline{d}}ayca&{\underline{d}}ayca&{\underline{d}}ayca& daxcdayca & daxcdayca\\
\cline{4-5}
{\underline{d}}axcdayca&{\underline{d}}axcdayca&{\underline{d}}axcdayca& dayca&
dayca\\
\hline
xcdayca& xcdayca & xcdayca & xcdayca & xcdayca \\
\hline
 yca & yca & yca & yca & yca \\
\hline
\end{tabular}
\end{table}

\subsection{Periods}
In lines 5 and 6 of the RadixSA pseudocode we detect periodic regions of the input  as follows:
if suffixes $i, i-p, i-2p, \ldots$ appear in the same bucket $b$, and bucket $b$ is currently sorted by
$d \ge p$ characters, then we have found a periodic region of the input,
where the period is $p$. If suffix $i$ is less than suffix $i+p$, then suffix $i-p$ 
is less than $i$, $i-2p$ is less than $i-p$, and so on. The case where
$i$ is greater than $i+p$ is analogous. \commentOut{We call suffix $i$ a period leader
because the order of the other suffixes in the period depends on suffix $i$. If
there are several periodic stretches in the same bucket we only sort the
period leaders, then we can easily induce the order of the other suffixes.}
Periods of any length are eventually detected because the depth of sorting in 
each bucket increases after each sort operation. 
This method can be viewed as a generalization of Seward's 
copy method \cite{Sew2000} where a portion of text of size $p$ is treated as a single character.

\subsection{Implementation Details}
Radix sorting is a central operation in RadixSA. We tried several implementations,
 both with Least Significant Digit (LSD) and Most Significant Digit (MSD) first order. 
The best of our implementations was a cache-optimized LSD radix sort. The cache optimization 
is the following. In a regular LSD radix sort, for every digit we
do two passes through the data: one to compute bucket sizes, one to assign
items to buckets. We can save one pass through the data per digit if in the bucket assignment pass we 
also compute bucket sizes for the next round \cite{LaMarca97}. We
took this idea one step forward and we computed bucket counts for all
rounds before doing any assignment. Since in our program we only sort numbers of
at most $64$ bits, we have a constant number of bucket size arrays to store in
memory. To sort small buckets we employ an iterative merge sort.

To further improve cache performance, in the bucket array we store not only the bucket start
position but also a few bits indicating the
length of the bucket. Since the bucket start requires $\lceil \log n
\rceil$ bits, we use the remaining bits, up to the machine word size, to
store the bucket length. This prevents a lot of cache misses when small buckets are the
majority. For longer buckets, we store the lengths in a separate array which also stores 
bucket depths.

\commentOut{
To reduce the size of the buffer array used in radix sort we first sort the data
by 16 bits, for which we don't need any buffer, just the count array and two
passes through the data. Then we take each bucket, buffer the next 64 bits of
each suffix in the bucket and apply radix sort. This completes step 1 in
the RadixSA pseudocode above.
}

The total additional memory used by the algorithm, besides input and output, is
 $5n + o(n)$ bytes: $4n$ for the bucket array, $n$ bytes for bucket depths and lengths, 
 and a temporary buffer for radix sort.

\section{Experimental Results}
One of the fastest SACAs,
in practice, is the Bucket Pointer Refinement (BPR) algorithm \cite{ScSt07}.
Version 0.9 of BPR has been compared \cite{ScSt07} with several other
algorithms: deep shallow \cite{MaFe04}, cache and copy by Seward \cite{Sew2000}, 
qsufsort \cite{LaSa07}, difference-cover \cite{BuKa03}, 
divide and conquer by Kim et al. \cite{KJP04}, and skew \cite{KaSa03}. BPR 0.9 has been
shown to outperform these algorithms on most inputs \cite{ScSt07}. 
Version 2.0 of BPR further improves over version 0.9. We compare RadixSA with both versions of 
BPR. 

Furthermore, a large set of SACAs are collected in the jSuffixArrays library \cite{JSA} 
under a unified interface. This library contains Java implementations of: DivSufSort \cite{DivSufSort} , 
QsufSort \cite{LaSa07}, SAIS \cite{NZC09}, skew \cite{KaSa03} and DeepShallow
\cite{MaFe04}. We include them in the comparison with the note that these 
Java algorithms may incur a performance penalty compared to their C counterparts. 

We tested all algorithms on an Intel core $i3$ machine with
$4GB$ of RAM, Ubuntu $11.10$ Operating System, Sun Java $1.6.0\_26$ virtual
machine and gcc $4.6.1$. The Java Virtual Machine was allowed to use up to $3.5GB$
of memory. As inputs, we used the datasets of \cite{ScSt07} which
include DNA data, protein data, English alphabet data, general ASCII alphabet data and
artificially created strings such as periodic and Fibonacci
strings\footnote{Fibonacci strings are similar to Fibonacci numbers, but addition is replaced with concatenation ($F_0=b, F_1=a, F_i$ is a concatenation of $F_{i-1}$ and $F_{i-2}$).}.

For every dataset, we executed each algorithm $10$ times. The average run times
are reported in table \ref{table_results} where the best run times are
shown in bold. \commentOut{For close values we used $T$-tests with a $p$-value of $0.01$ to 
confirm or reject ties.} Furthermore, we counted the number of times RadixSA accesses each suffix.
The access counts are shown in figure \ref{fig_avg_access}. For almost all datasets, the number
of times each suffix is accessed is a small constant. For the Fibonacci string the
number of accesses is roughly logarithmic in the length of the input. 

\begin{figure}[h]
\caption{\label{fig_avg_access}Average number of times RadixSA accesses each suffix, for datasets from
\cite{ScSt07}.}
\includegraphics[scale=0.4]{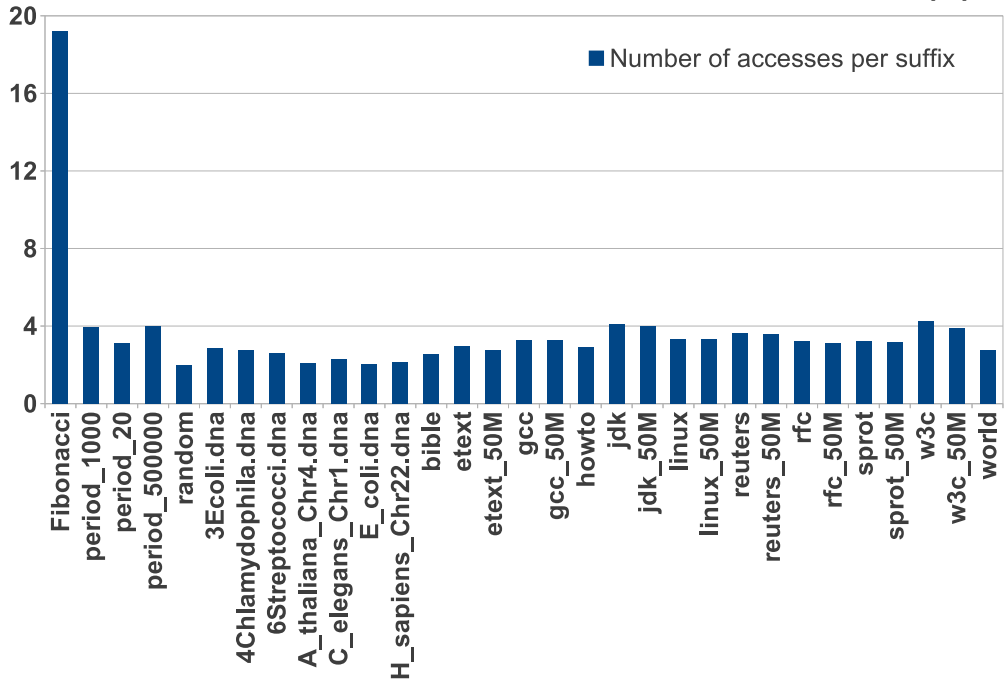}
\end{figure}

    \begin{table}
    
    \caption{\label{table_results}Comparison of run times on datasets from \cite{ScSt07} on a 64-bit Intel CORE i3
machine with 4GB of RAM, Ubuntu $11.10$ Operating System, Sun Java $1.6.0\_26$
and gcc $4.6.1$. Run times are in seconds, averaged over 10 runs. Bold font indicates the
best time. \commentOut{Ties were established
using T-tests.} ML means out of memory, TL means more than 1 hour.
    }
    \scalebox{0.75}{
\begin{tabular}{| c | c | c | c | c | c | c |
c | c | c | c |}
\hline 
\multicolumn{3}{| c |}{Dataset} & \multicolumn{8}{ c |}{Run time}  \\
\hline
 &  & & & & & DivSuf &
QSuf & &  & Deep \\
Name & Length & $|\Sigma|$ & RadixSA  & BPR2  & BPR.9 &  Sort &
Sort & SAIS & skew & Shallow 
\\
\hline
Fibonacci & 20000000 & 2 & 7.88 & 12.48 & 14.05 & { 6.81} & 26.44 & {\bf
5.50} & 14.53 & 369.48
\\
\hline
period\_1000 & 20000000 & 26 & {\bf 2.12} & 3.52 & 5.71 &  { 3.15} & 20.42 &
6.59 & 23.27 & TL
\\
\hline
period\_20 & 20000000 & 17 & {\bf 1.44} & 1.95 & 43.39 & { 1.83} & 11.05 &
2.83 & 7.15 & TL
\\
\hline
period\_500000 & 20000000 & 26 & {\bf 2.78} & { 4.60} & 6.31 & 4.74 & 23.32 &
8.56 & 25.68 & 2844.37
\\
\hline
random & 20000000 & 26 & {\bf 2.25} & { 3.34} & 4.87 & 6.35 & 5.02 & 11.75 &
22.05 & 5.69
\\
\hline
3Ecoli.dna & 14776363 & 5 & {\bf 2.23} & { 2.67} & 3.43 & 4.00 & 13.85 & 6.14
& 19.62 & 433.54
\\
\hline
4Chlamydophila.dna & 4856123 & 6 & {\bf 0.61} & { 0.67} & 0.90 & 1.71 & 3.24
& 1.93 & 5.24 & 4.80
\\
\hline
6Streptococci.dna & 11635882 & 5 & {\bf 1.63} & { 1.79} & 2.38 & 2.88 & 7.08
& 4.98 & 14.88 & 4.26
\\
\hline
A\_thaliana\_Chr4.dna & 12061490 & 7 & {\bf 1.27} & { 1.74} & 2.40 & 3.02 &
5.13 & 5.37 & 15.71 & 3.52
\\
\hline
C\_elegans\_Chr1.dna & 14188020 & 5 & {\bf 1.61} & { 1.95} & 2.65 & 3.21 &
6.91 & 5.69 & 17.18 & 6.92
\\
\hline
E\_coli.dna & 4638690 & 4 & {\bf 0.41} & { 0.51} & 0.58 & 1.36 & 1.72 & 1.96
& 5.04 & 1.37
\\
\hline
H\_sapiens\_Chr22.dna & 34553758 & 5 & {\bf 4.40} & { 5.66} & 8.21 & 7.76 &
15.31 & 15.59 & 49.70 & 10.98
\\
\hline
bible & 4047391 & 63 &  { 0.51} & {\bf 0.48} & 0.80 & 1.24 & 1.38 &
1.56 & 4.64 & 1.08
\\
\hline
etext & 105277339 & 146 & {\bf 19.40} & { 23.09} & 43.46 & 26.56 & 62.63 &
54.70 & ML & 119.96
\\
\hline
etext\_50M & 50000000 & 120 & {\bf 8.13} & { 9.74} & 17.26 & 11.94 &
26.40 & 24.46 & 88.57 & 79.07
\\
\hline
gcc & 86630400 & 150 & {\bf 13.84} & { 15.58} & 24.50 & 15.84 & 46.20 & 33.62
& 135.12 & 80.78
\\
\hline
gcc\_50M & 50000000 & 121 & {\bf 7.21} &  9.56 & 13.26 & { 8.31} & 28.43 &
17.73 & 68.65 & 264.90
\\
\hline
howto & 39422104 & 197 & {\bf 5.96} & { 6.35} & 10.26 & 8.41 & 17.64 & 16.67
& 64.73 & 16.33
\\
\hline
jdk & 69728898 & 113 & {\bf 12.07} & { 12.54} & 26.86 & { 12.74} & 39.92 &
24.66 & 102.76 & 58.22
\\
\hline
jdk\_50M & 50000000 & 110 & {\bf 8.32} & {\bf 8.30} & 17.05 & { 8.91} & 26.30
& 17.58 & 71.31 & 36.98
\\
\hline
linux & 116254720 & 256 & {\bf 19.27} & {\bf 19.34} & 29.67 & { 21.17} &
61.99 & 44.47 & ML & 58.71
\\
\hline
linux\_50M & 50000000 & 256 & {\bf 7.62} & {\bf 7.60} & 10.50 & { 8.84} &
27.54 & 18.18 & 76.10 & 31.92
\\
\hline
reuters & 114711150 & 93 & {\bf 19.76} & { 25.08} & 60.72 & { 25.07} &
74.78 & 49.17 & ML & 87.57
\\
\hline
reuters\_50M & 50000000 & 91 & {\bf 7.84} & { 9.53} & 20.41 & 10.24 & 26.94 &
20.29 & 77.25 & 33.68
\\
\hline
rfc & 116421900 & 120 & {\bf 21.18} & { 22.08} & 42.75 & 22.55 & 66.28 &
47.99 & ML & 42.14
\\
\hline
rfc\_50M & 50000000 & 110 & {\bf 8.23} & { 8.39} & 14.85 & 9.24 & 24.80
& 19.61 & 76.64 & 16.63
\\
\hline
sprot & 109617186 & 66 & {\bf 18.48} & { 22.79} & 47.07 & 25.52 & 69.58 &
50.40 & ML & 48.69
\\
\hline
sprot\_50M & 50000000 & 66 & {\bf 7.57} & { 9.10} & 16.81 & 10.88 & 28.07 &
21.69 & 78.47 & 20.03
\\
\hline
w3c & 104201578 & 256 & {\bf 18.82} & {\bf 18.78} & 35.94 & { 20.01} & 74.09
& 38.29 & ML & 1964.80
\\
\hline
w3c\_50M & 50000000 & 255 & {\bf 7.93} & { 8.33} & 17.67 & 8.73 & 25.95 &
17.26 & 71.42 & 36.59
\\
\hline
world & 2473399 & 94 & { 0.30} & {\bf 0.27} &  0.42 & 0.91 & 0.86 & 0.91
& 2.35 & 0.78
\\
\hline
\end{tabular}
}
    \end{table}  

\section{Discussion and Conclusions}

In this paper we have presented an elegant algorithm for the construction of
suffix arrays. This algorithm is one of the simplest algorithms known for
suffix arrays construction and runs in $O(n)$ time on a large fraction of all
possible inputs. It is also nicely parallelizable. We have shown how our
algorithm can be implemented on various parallel models of computing. 

We have also given an extension of this algorithm, called RadixSA, which has a worst
case runtime of $O(n \log n)$ and proved to be efficient in practice. 
RadixSA uses a heuristic to select the order in which buckets are 
processed so as to reduce the number of operations performed. RadixSA performed a linear number 
of operations on all but one 
of the inputs tested. The heuristic could find application as an independent speedup technique 
for other algorithms which use bucket sorting and induced copying. For example, BPR could
use it to determine the order in which it chooses buckets to be refined.
A possible research direction is to improve RadixSA's heuristic. Buckets can be processed based on 
a topological sorting of their dependency graph. Such a graph has at most $n/2$ 
nodes, one for each non singleton bucket, and at most $n/2$ edges. Thus, it
has the potential for a lightweight implementation. 
 
An interesting open problem is to devise a randomized algorithm that has a
similar performance.

\section{Authors contributions}
SR and MN designed and analyzed the algorithms. MN implemented RadixSA
and carried out the empirical experiments. SR and MN analyzed the empirical
results and drafted the manuscript.

\section{Acknowledgements}
 This work has been supported in part by the following grants: NSF 0829916 and NIH R01LM010101.



\bibliographystyle{elsarticle-num}
\bibliography{sa-references} 









\end{document}